\documentclass[12pt]{amsart}
\usepackage{fullpage,url,amssymb,enumerate,colonequals}

\usepackage{todonotes}
\usepackage{enumitem}

\usepackage{color}

\usepackage[
        colorlinks, citecolor=darkgreen,
        backref,
        pdfauthor={}, 
]{hyperref}

\usepackage{comment}

\newcommand{\C}{\mathbb{C}}

\newcommand{\Z}{\mathbb{Z}}

\DeclareMathOperator{\Prob}{\bf P}

\newcommand{\GL}{\operatorname{GL}}

\numberwithin{equation}{section}

\newtheorem{theorem}[equation]{Theorem}
\newtheorem{lemma}[equation]{Lemma}

\newtheorem{proposition}[equation]{Proposition}

\theoremstyle{definition}
\newtheorem{definition}[equation]{Definition}

\theoremstyle{remark}
\newtheorem{remark}[equation]{Remark}

\definecolor{darkgreen}{rgb}{0,0.5,0}

\begin{document}

\title{The Dihedral Hidden Subgroup Problem}

\author{Imin Chen}

\address{Department of Mathematics, Simon Fraser University\\
Burnaby, BC V5A 1S6, Canada } \email{ichen@sfu.ca}

\author{David Sun}

\address{Department of Mathematics, Simon Fraser University\\
Burnaby, BC V5A 1S6, Canada } \email{david\_sun\_2@sfu.ca}

\date{\today}

\subjclass[2010]{}

\thanks{This work was supported by a NSERC Discovery Grant (I.C.) and NSERC USRA (D.S.).}

\begin{abstract}
We give an overview of the dihedral hidden subgroup problem (DHSP) as approached by the `standard' hidden subgroup quantum algorithm for finite groups, highlighting the obstructions for strong Fourier sampling to succeed, and summarizing known approaches and results.

We then prove a number of no-go results for the dihedral coset problem (DCP), motivated by a connection between DCP and cloning of quantum states.
\end{abstract}

\maketitle

\tableofcontents

\section{Introduction}

Let $G$ be a finite group and $H$ a hidden subgroup of $G$.  A function $f : G \rightarrow \C$ which is constant on left $H$-cosets and takes distinct values on distinct left $H$-cosets is called a separating function for the subgroup $H$.

The hidden subgroup problem (HSP) is the problem of finding generators for the hidden subgroup $H$, given access to evaluations of a separating function $f$ for $H$. This problem can be solved in polynomial time using a quantum computer when $G$ is an abelian group and has been extensively studied for many classes of finite groups \cite{kitaev} \cite{GSVV} \cite{HRT}. 

Many problems can be cast in terms of the HSP. For instance, Shor's integer factorization algorithm can be described in terms of the HSP for cyclic groups \cite{kitaev} \cite{shor}. A polynomial time quantum algorithm for solving the hidden subgroup problem on dihedral groups would imply a polynomial time quantum algorithm to solve certain hard lattice problems which are considered intractable using classical computers \cite{regev}. Another example is the HSP on the symmetric group which can be used to solve the graph isomorphism problem \cite{beals} \cite{boneh-lipton} \cite{ettinger-hoyer-4} \cite{hoyer}.

Though the dihedral group is one of the simplest non-abelian groups, from the point of view of the HSP, it has remained a difficult case. A survey of known results about the HSP for dihedral groups can be found in \cite{kobayashi} \cite{HRT} \cite{childs}, where we note that the best known quantum algorithms are currently subexponential \cite{kuperberg} \cite{regev-2} \cite{kuperberg-2}. 

The latter subexponential algorithms have applications to constructing isogenies between elliptic curves over finite fields \cite{childs-jao-soukharev} \cite{biasse-jao-sankar}. In \cite{BKSW}, it is shown that the learning with errors problem (LWE) is quantum polynomial time equivalent to an extrapolated version of the dihedral coset problem. Both the problem of constructing isogenies between supersingular elliptic curves and LWE form the basis for many proposed post-quantum key exchanges, therefore the quantum hardness of the HSP for groups like the dihedral group becomes a critical question.

In this paper, we review the standard HSP algorithm as it applies to the dihedral groups $D_N$ and detail the obstructions for this algorithm to succeed in this case. On the other hand, we explain how the standard HSP algorithm yields the polynomial query complexity result of \cite{ettinger-hoyer}. We also describe other approaches to the HSP for dihedral groups, such as optimal measurements and its relations to the subset sum problem. 

We then prove a number of no-go theorems for the dihedral coset problem (DCP). The results yield an upper bound on the success parameter of any quantum algorithm which uses a unitary operation and then one measurement to determine the parameter $a$ with bounded probability. This can be viewed as giving a non-trivial upper bound on the success probability of the optimal measurement using $m$ coset samples to solve DCP in the case when the density $\nu = m/\log_2 N \ge 1$ and the order of the dihedral group is $2N$. 

In the last section, we describe a connection between DCP and cloning of quantum states which helped motivate the proofs of the no-go results for DCP. 

\section{Acknowledgements}

We would like to thank P.\ H\o yer for helpful comments and bringing to our attention \cite{kuperberg}. We also thank R.\ Goenka and N.\ de Silva for stimulating discussions related to the topics of this paper.

\section{The QFT for finite groups}

Let $G$ be a finite group and $\hat G$ denote a  complete set of representatives for the isomorphism classes of irreducible representations of $G$ over $\C$. For a representation $\rho \in \hat G$, let $d_\rho$ be the dimension of $\rho$. Recall the Quantum Fourier Transform (QFT) on $G$ is defined as the linear transformation
\begin{align}
  F_G & : V \rightarrow \hat V \\
  \notag F_G & := \sum_{g \in G} \sum_{\rho \in \hat G} \sum_{i,j=1}^{d_\rho} \sqrt{\frac{d_\rho}{|G|}} \rho(g)_{i,j} | \rho, i, j \rangle \langle g |,
\end{align}
where $V$ is the $\C$-vector space generated by $| g \rangle$, $g \in G$ and $\hat V$ is the $\C$-vector space generated by $| \rho, i, j \rangle$, $\rho \in \hat G, 1 \le i, j \le d_\rho$. Picking an isomorphism $V \cong \hat V$, it is a unitary operator which can be efficiently approximated using quantum circuits for many finite groups: \cite{hales-hallgren} (abelian), \cite{hoyer} (meta-cyclic), \cite{beals} (symmetric group).

\subsection{The cyclic group case} Suppose $G = C_N \cong \Z/N\Z$ the cyclic group of order $N$.  There are $N$ irreducible representations which are $1$-dimensional and given by 
\begin{align*}
  \Z/N\Z & \rightarrow \C^\times \\
  t & \mapsto \zeta_N^t.
\end{align*}
where $\zeta_N$ is a choice of $N$th root of unity.

\subsection{The dihedral group case}

Suppose $G = D_N$ is the dihedral group of order $2n$, which can be presented as
\begin{equation*}
  D_N = \langle x, y : x^n = e, y^2 = e, y x y^{-1} = x^{-1} \rangle.
\end{equation*}

If $n$ is even,  there are four $1$-dimensional representations given by
\begin{equation}
  \phi_{u,v} : x \mapsto (-1)^u, y \mapsto (-1)^v 
 \end{equation}
where $u, v \in \Z/2\Z$. These are pull backs of the four $1$-dimensional representations of $D_N/\langle x^{2} \rangle \cong C_2 \times C_2$ under the quotient homomorphism $D_N \rightarrow D_N/\langle x^{2} \rangle$, where $C_m$ denotes the cyclic group of order $m$.

If $n$ is odd, there are two $1$-dimensional representations given by $\phi_{0,v}$  where $v \in \Z/2\Z$. These are pull backs of the two $1$-dimensional representations of $D_N/\langle x \rangle \cong C_2$ under the quotient homomorphism $D_N \rightarrow D_N/\langle x \rangle$.

There are $\lfloor \frac{n-1}{2} \rfloor$ irreducible representations of dimension $2$ given by
\begin{align}
\label{complex-basis}
  \rho_k: D_N & \rightarrow \GL_2(\C) \\
  \notag x & \mapsto \begin{pmatrix} \omega_N^k & 0 \\ 0 & \omega_N^{-k} \end{pmatrix} \\
  \notag y & \mapsto \begin{pmatrix} 0 & 1 \\ 1 & 0 \end{pmatrix}
\end{align}
for $0 < k < \frac{n}{2}$, where $\omega_N = e^{2 \pi i/n}$. These are the induction of the representation $\psi_k : C_{n} \rightarrow \C^\times$ given by $\psi_k(x) = \omega_N^k$ from $C_n$ to $D_N$. 

The representations $\phi_{u,v}$ and $\rho_k$ form the complete list of irreducible representations of $D_N$ up to isomorphism.

\section{The standard HSP algorithm}

In the standard algorithm for finding hidden subgroups from a separating function, we perform the following steps:

We form the state
\begin{equation}
\label{uniform-state}
   \sqrt{\frac{1}{|G|}} \sum_{g \in G} |g \rangle | f(g) \rangle,
\end{equation}
where $f : G \rightarrow \C$ is the given separating function. 

This can be achieved by starting with the state $| e_G \rangle | 0 \rangle$, where $e_G$ is the identity element of $G$, then performing the following computations:
\begin{align}
  | e_G \rangle | 0 \rangle & \mapsto \sqrt{\frac{1}{|G|}} \sum_{g \in G} | g \rangle | 0 \rangle \quad \text{Apply the QFT over $G \leftrightarrow \Z/|G|\Z, e_G \leftrightarrow 0$ to first register.} \\
  & \mapsto \sqrt{\frac{1}{|G|}} \sum_{g \in G} |g \rangle | f(g) \rangle \quad \text{ Compute $f$ into second register.}
\end{align}

Measuring the second register and discarding it, we obtain a state of the form
\begin{equation}
\label{coset-samples}
  \sqrt{\frac{1}{|H|}} \sum_{h \in H} | ch \rangle. 
\end{equation}

We apply the QFT to the above state to obtain
\begin{equation}
   \sum_{\rho \in \hat G} \sqrt{\frac{d_\rho}{|G||H|}} \sum_{i,j =1}^{d_\rho}  \sum_{h \in H}  \rho(ch)_{i,j} | \rho, i, j \rangle.
\end{equation}

In the case of $G$ being an abelian group, measuring $\rho$ gives sufficient information to determine $H$ efficiently after running this process repeatedly and using post-processing \cite{kitaev}.

\subsection{The cyclic group case}

Fix an integer $N > 1$. Let $X$ be a finite set, and $G := (\mathbb{Z}/N\mathbb{Z}, +)$. Suppose that we have a function $f:G \rightarrow X$ which separates a subgroup $H \subset G$ where $H = \langle d \rangle$. Let $M := \#H$. Assume that we have a quantum machine capable of computing the unitary transformation on two registers $U_f: |x \rangle|y \rangle \rightarrow |x \rangle|f(x) \oplus y \rangle$ (recall that we can take $|x \rangle|y \rangle$ as $|x \rangle \otimes |y \rangle$).

Suppose we do not know $M$, $d$ nor $H$ and we only know $G$ and have a machine computing $f$. We want to determine a generating set for $H$, calling the "black-box" function $f$ as few times as possible. 
	
Let $F_N$ be the QFT for the cyclic group $G$. Explicitly, this is an operator on a register with $n \ge \log_2N$ qubits given by $$F_N := \frac{1}{\sqrt{N}}\sum_{j,k=0}^{N-1}\exp\left( \frac{2\pi ijk}{N}\right) |k \rangle \langle j |.$$

The $F_N$ is unitary transformation. If we let $\omega := \exp\left(\frac{2\pi i}{N}\right)$ be the primitive $N$-th root of unity, then
	$$F_N = \frac{1}{\sqrt{N}}\begin{pmatrix} 
		1 & 1 & 1 & \cdots & 1\\
		1 & \omega & \omega^2 & \cdots & \omega^{N-1}\\
		1 & \omega^2 & \omega^4 & \cdots & \omega^{2(N-1)}\\
		\vdots & \vdots & \vdots & \; & \vdots\\
		1 & \omega^{N-1} & \omega^{2(N-1)} & \cdots & \omega^{(N-1)(N-1)}\\
	\end{pmatrix}.$$
	One can check that $F_N\cdot F_N^{*} = I_{N}$ where $I_N$ is the $N \times N$ identity matrix.

We map $G = \{0,1,...,N-1\}$ onto the basis of the quantum state $\{|0 \rangle,|1 \rangle,...,|N-1 \rangle\}$. Suppose the hidden subgroup is given by $H = \{|0 \rangle, |d \rangle, |2d \rangle,...,|(M-1)d \rangle\}$. 

Computing on two registers:
\begin{flalign*}
	|0 \rangle|0 \rangle & \; \xrightarrow{\text{$F_N$ on 1st}} \; \frac{1}{\sqrt{N}}\sum_{j=0}^{N-1}|j \rangle|0 \rangle\\
	&\; \xrightarrow{\text{apply $f$}} \;\;\;\; \frac{1}{\sqrt{N}}\sum_{j=0}^{N-1}|j \rangle|f(j) \rangle.
\end{flalign*}
Note that we put $|f(j) \rangle$ inside the sum since tensor product is distributive. Measuring in $|f(j_0) \rangle$ on the second register for some $0 \leq j_0 \leq N-1$ collapses our state, leaving only those values $g \in G$ such that $f(g) = f(j_0)$ in the first register. Since $f$ separates cosets of $H$ we get (for simplicity, we now drop our second register which remains $|f(j_0) \rangle$):
\begin{flalign*}
	&\xrightarrow{\text{measure}} \; \frac{1}{\sqrt{M}}\sum_{h \in H}|j_0 + h \rangle = \; \frac{1}{\sqrt{M}}\sum_{s = 0}^{M-1}|j_0 + sd \rangle\\
	&\xrightarrow{\text{apply $F_N$}} \; \frac{1}{\sqrt{M}}\sum_{s=0}^{M-1}\frac{1}{\sqrt{N}}\sum_{k=0}^{N-1}\exp\left(\frac{2\pi i(j_0 + sd)k}{N} \right) |k \rangle\\
	&= \; \frac{1}{\sqrt{MN}}\sum_{k=0}^{N-1}\exp\left(\frac{2\pi i j_0k}{N}\right)|k \rangle \sum_{s=0}^{M-1}\exp\left(\frac{2\pi i sdk}{N} \right) \\
	& =  \frac{1}{\sqrt{d}}\sum_{t=0}^{d-1}\exp\left(\frac{2\pi i j_0 tM}{N} \right) |tM \rangle,
\end{flalign*}
using the fact that
$$\sum_{s=0}^{M-1}\exp\left(\frac{2\pi i sdk}{N} \right) = \sum_{s=0}^{M-1}\exp\left(\frac{2\pi i k}{M} \right)^s = \begin{cases}
	0, \;\;\; \text{ if $M \nmid k$},\\
	M, \; \text{ if $M \mid k$}
\end{cases}$$
for $0 \leq k \leq N-1$ and that $\frac{M}{N} = \frac{1}{d}$.

Now measurement at this point gives a multiple of $M$ in $\{0,M,...,(d-1)M\}$ with uniform probability. We repeat this whole process many times to obtain a collection of multiples of $M$ and take the GCD to obtain $M$ with high probability.

To estimate how many trials $m \ge 2$ we need, suppose we have $t_1,...,t_m \in \{0,1,...,d-1\}$. We want to estimate the probability that $\gcd(t_1,...,t_m) = 1$, in particular we have the lower bound
\begin{equation}
  \Prob(\gcd(t_1,...,t_m) = 1) \ge \zeta(m)^{-1} + O(\log d/d),
\end{equation}
where $\zeta(s)$ is the Riemann zeta function by \cite{nymann}. Thus a few runs of this algorithm determines $H$ with high probability for any $N$ and `most' $d$.

\begin{lemma}\label{inverse_zeta_bound}
 We have that 
\begin{equation*} 
 \zeta(m)^{-1} > 1 - 3 \cdot 2^{-m}.
\end{equation*}
for every $m \ge 2$.
\end{lemma}
\begin{proof}
We first recall that
\begin{equation*}
\zeta(s) = \sum_{n=1}^{\infty}n^{-s} \text{\;\;\; and \;\;\;} \zeta(s)^{-1} = \sum_{n=1}^{\infty}\frac{\mu(n)}{n^s}.
\end{equation*}
for real $s \ge 2$ where $\mu$ is the M\"{o}bius function. Then
\begin{align*}
  \frac{1 - \zeta(s)^{-1}}{2^{-s}} & = \sum_{n \ge 2} \frac{-\mu(n)}{(n/2)^s} \\
  & \le \sum_{n \ge 2} \frac{1}{(n/2)^2} \\
  & = 4 (\zeta(2) - 1) < 3.
\end{align*}
\end{proof}

We may therefore view the standard algorithm for HSP on the cyclic group $G$ as producing a quantum state of the form
\begin{equation}
    \sum_{t_1, \ldots, t_m} \alpha_{t_1, \ldots, t_m} |t_1 M \rangle \ldots | t_m M \rangle
\end{equation}
We may compute the greatest common divisor of the above registers into a blank register:
\begin{equation}
    \sum_{t_1, \ldots, t_m} \alpha_{t_1, \ldots, t_m} |t_1 M \rangle \ldots | t_m M \rangle | 0 \rangle \mapsto \sum_{t_1, \ldots, t_m} \alpha_{t_1, \ldots, t_m} |t_1 M \rangle \ldots | t_m M \rangle | \gcd(t_1 M, \ldots t_m M) \rangle.
\end{equation}
Thus, the standard HSP algorithm for $G$ can be viewed as a unitary operation of the form:
\begin{equation}
\label{cyclic-HSP}
    |A \rangle | \psi_d^1 \rangle \cdots | \psi_d^m \rangle  | 0 \rangle \mapsto \sum_{e} | \Sigma_{e} \rangle | N/e \rangle \mapsto \sum_{e} | \Sigma_{e} \rangle |e \rangle,
\end{equation}
satisfying 
\begin{equation*}
  | \Sigma_d |^2 \ge \lambda(m,d),
\end{equation*}
for every $m$. We remark the second map sending $e \mapsto N/e$ in the last register is unitary (if $e \nmid N$, the map leaves $e$ alone).

\begin{remark}
\label{cyclic-go}
Assume for any guess for $d$, there is a quantum circuit which can decide if $d$ is correct. For a fixed $m$, we can improve the success probability above by the following method. Let us instead consider the probability of achieving a multiple $k M$ of $M$ for $1 \le k \le C$ for some $C \in \mathbb{N}$. For the given guess of $M$ and hence for $d$, we can check if it is the correct value, and if not, adjust it to the correct value because we know the true value is $d/k$ for some $1 \le k \le C$ and $d/k$ being an integer. This increases the success parameter for a fixed number of samples $m \ge 2$. For instance, if $m = 2$ and $C = 10$, then the success parameter improves from $\approx 0.6079$ to $\approx 0.9892$. 
\end{remark}

The above example motivates the next definition.
\begin{definition}
\label{find-parameter}
Let $\mathcal{I}_d = \left\{ |\psi_d \rangle \right\}$ a collection of possible input states with length $N$ and depending on a parameter $d$. The problem of determining $d$ from 
a list of $m$ samples in $\mathcal{I}_d$ is {\it unitarily solvable with success parameter $\lambda(m,d)$} if there is a unitary operator which has the effect
\begin{equation}
\label{main-unitary}
     |A \rangle | \psi_d^1 \rangle \cdots | \psi_d^m \rangle | 0 \rangle \mapsto |\Sigma_d(\psi_d) \rangle | d \rangle  + \sum_{e \not= d}  |\Sigma_e(\psi_d) \rangle | e \rangle,
\end{equation}
 where
\begin{align*}
   & | \Sigma_d(\psi_d) |^2 \ge \lambda(m,d),
\end{align*}
for every $m$ and $d$. 
\end{definition}

We may view \eqref{main-unitary} as computing a main term 
\begin{equation}
    | \Sigma_d(\psi_d) \rangle | d \rangle
\end{equation}
with error term
\begin{equation}
  \sum_{e \not= d} |\Sigma_e(\psi_d) \rangle | e \rangle.
\end{equation}


The next theorem is stated for completeness and for later comparison to the case of DCP. It summarizes the well-known standard algorithm for HSP on a finite cyclic group in terms of the definitions above.
\begin{theorem}
\label{finite-cyclic-algorithm}
  The problem of determining a generator for a hidden subgroup of a finite cyclic group, given a list of $m$ HSP coset samples, is unitarily solvable with success parameter $\zeta(m)^{-1} + O(\log d/d)$ where $\zeta(s)$ is the Riemann zeta function. 
\end{theorem}

\begin{remark}
  Here $M = N/d$, so we may view the standard quantum algorithm as producing uniform samples in $H^\perp = \langle d \rangle^\perp = \langle N/d \rangle$. For a general abelian group $G$, the uniform samples in $H^\perp$ from the standard quantum algorithm are used to determine $H$ using a classical probabilistic algorithm.
\end{remark}

\subsection{The dihedral group case}

In \cite{ettinger-hoyer}, it is shown that the hidden subgroup problem for $G = D_N$ for a general subgroup $H$ is reduced to the case of a single reflection subgroup $H = H_a$. 

For $H = H_a =  \langle y x^a \rangle$, the probability of obtaining $| \rho, i, j \rangle$ is $\frac{1}{|G|}$ when $d_\rho = 2$, which does not allow one to distinguish the groups $H_a$. Explicitly, in the complex basis \eqref{complex-basis}:

If $\rho = \rho_k$,  then
\begin{align}
   \label{dcp-case}
   \sum_{h \in H}  \rho(x^\alpha h) & = \begin{pmatrix}
   \omega_N^{\alpha k} & \omega_N^{-(a-\alpha)k} \\
   \omega_N^{(a-\alpha)k} & \omega_N^{-\alpha k}
  \end{pmatrix} \\
    \sum_{h \in H}  \rho(y x^\alpha h) & = \begin{pmatrix}
   \omega_N^{(a-\alpha)k} & \omega_N^{-\alpha k} \\
   \omega_N^{\alpha k} & \omega_N^{-(a-\alpha)k} \\
  \end{pmatrix}.
\end{align}

If $\rho = \phi_{u,v}$, then
\begin{align}
  \sum_{h \in H} \rho(x^\alpha h) & = (-1)^{\alpha u} + (-1)^{v + (a-\alpha)u} \\
  \notag & = (-1)^{\alpha u} (-1 + (-1)^{v + au}) \\
   \sum_{h \in H} \rho(y x^\alpha h) & = (-1)^{(a-\alpha) u} + (-1)^{v + \alpha u} \\
  \notag & = (-1)^{(a-\alpha) u} (-1 + (-1)^{v + au}).
\end{align}

If one changes to the real basis, we get a probability distribution dependent on $a$, but it is very flat, making it hard to distinguish the subgroups $H_a$.

More generally, in order for the QFT to be an unitary operator, we require that $\rho_k(g)$ be unitary for every $k$ and $g \in D_N \implies |\rho_k(g)_{i,j}| \le 1$ for $1 \le i, j \le 2$. In particular,
for any set of $2$-dimensional irreducible representations $\rho_k$, we have that
\begin{align}
\label{fourier-bound}
  \Prob(\rho_k,i,j) &= \frac{1}{n} | (\rho_k(c y x^a) + \rho_k(c))_{i,j} |^2 \\
  \notag & \le \frac{1}{N} ( | \rho_k(c y x^a)_{i,j}| + | \rho_k(c)_{i,j}| )^2 \\
  \notag & \le \frac{4}{N},
\end{align}
where $\Prob(\rho_k,i,j)$ is the probability of observing the state $| \rho_k, i,j \rangle$.  Although the choice of basis may result in probability distributions of states which depend on $a$, if $N$ is very large, the above inequalities show that the probabilities will always be very flat.

In \cite{moore-russell}, it is shown that the optimal positive operator valued measurement (POVM) to determine $a$ from a single DCP sample exists and is given by the the pretty good measurement (PGM). Also, the optimal measurement has success probability
\begin{equation}
\label{pgm-bound}
    \Prob_\text{success} = \frac{2}{N} \left( 1 - \frac{1}{2N} \right).
\end{equation}

\begin{theorem}
  The standard algorithm for DHSP cannot implement the optimal measurement using one coset sample. 
\end{theorem}
\begin{proof}
  This follows because \eqref{fourier-bound} and \eqref{pgm-bound} are incompatible.
\end{proof}

\section{Dihedral coset sampling}

In the standard HSP algorithm, after the first step we are left with random coset samples as in \eqref{coset-samples}. In the case of $G = D_N$, the dihedral group of order $n$, and $H = H_a = \langle y x^a \rangle$, this is explicitly of the form
\begin{align}
      \frac{1}{|H|} \sum_{h \in H} | ch \rangle & = \frac{1}{\sqrt{2}} \left( | c \rangle + | c y x^a \rangle \right) \\
      \notag & = \frac{1}{\sqrt{2}} \left( | y^\beta x^\alpha \rangle + | y^\beta x^\alpha y x^a \rangle \right) \\
      \notag & = \begin{cases}
      		\frac{1}{\sqrt{2}} \left( | x^\alpha \rangle + | y x^{a-\alpha} \rangle \right) & \text{ if } \beta = 0\\
      		\frac{1}{\sqrt{2}} \left( | x^{a-\alpha} \rangle + | y x^{\alpha} \rangle \right) & \text{ if } \beta = 1\\
	\end{cases}
\end{align}
where $c = y^\beta x^\alpha$. 

\begin{remark}
\label{hsp-dcp}
The second case is reduced to the first by the transformation $\alpha \rightarrow a - \alpha$ if this transformation leaves the distribution of $\alpha$ invariant.
\end{remark}

Given samples of the form
\begin{equation}
\label{unified-coset-sample}
  \psi_a = \psi_{a; \alpha} = \frac{1}{\sqrt{2}} \left( | x^\alpha \rangle + | y x^{a-\alpha} \rangle \right),
\end{equation}
the dihedral coset problem (DCP) is the problem of finding generators for the hidden subgroup $H = H_a$. The states $\psi_a = \psi_{a; \alpha}$ are called DCP samples for $a$.

For HSP samples produced from the standard algorithm, where $\alpha$ is from the uniform distribution, we may view HSP samples as DCP samples by Remark~\ref{hsp-dcp}.

\begin{remark}
We can encode a DCP sample $\psi_{a;\alpha}$ as 
\begin{equation*}
  \frac{1}{\sqrt{2}} ( | 0 \rangle | \alpha  \rangle + |1 \rangle | a - \alpha \rangle).
\end{equation*}
Using the fact that $y x^\alpha = x^{-\alpha} y$, this can be encoded (after negating $a$) as
\begin{equation*}
  \frac{1}{\sqrt{2}} ( | 0 \rangle | \alpha  \rangle + |1 \rangle | a + \alpha \rangle),
\end{equation*}
which is another commonly used form used in the literature, especially in the context of the `hidden shift problem'.
\end{remark}


\section{Other approaches to DHSP and DCP}

\subsection{Subexponential algorithms} The first row of \eqref{dcp-case} can be encoded as
\begin{align*}
    & \frac{1}{\sqrt{2N}} \sum_{k} \left( \omega_N^{\alpha k}  | k \rangle |0 \rangle + \omega_N^{(a-\alpha)k}  | k \rangle |1 \rangle \right) \\
    & = \frac{1}{\sqrt{N}} \sum_{k} \omega_N^{\alpha k} \otimes \frac{1}{\sqrt{2}} \left( |0 \rangle + \omega_N^{a k} | 1 \rangle \right).
\end{align*}
Measuring the first register yields samples of the form
\begin{align}
\label{tensor-sample}
  | \Psi_k \rangle = \frac{1}{\sqrt{2}} \left( |0 \rangle + \omega_N^{a k} | 1 \rangle \right),
\end{align}
where $k$ is known from the measurement. 

Let $N = 2^t$ for simplicity and $m = \lceil \sqrt{t-1} \rceil$. The idea behind the subexponential algorithm in \cite{kuperberg} is to combine states of the form \eqref{tensor-sample}. In particular, we see that
\begin{align}
\label{tensor-combine}
    | \Psi_p \rangle | \Psi_q \rangle = \frac{1}{\sqrt{2}} \left( | \Psi_{p+q} \rangle |0 \rangle + \omega_N^{aq} | \Psi_{p-q} \rangle | 1 \rangle \right).
\end{align}
If $p$ and $q$ have the same $mj$ least significant bits, then $p \pm q$ strictly increases the number of least significant bits $p$ and $q$ share.

With sufficiently many samples of the form $\Psi_p$ which have $mj$ common least significant bits, it is shown in \cite{kuperberg} that combining the states as in \eqref{tensor-combine} produces enough states with $m(j+1)$ common least significant bits. Thus, sieving from enough samples at the outset, we eventually produce states of the form
\begin{equation*}
   \Psi_{2^{t-1}} = | 0 \rangle + (-1)^a | 1 \rangle
\end{equation*}
which are sufficient to determine the parity of $a$. It is shown in \cite{kuperberg} that the above method yields an algorithm which requires $2^{O\left(\sqrt{\log N}\right)}$ time, space, and queries. In \cite{regev}, a modified algorithm is given which requires $2^{O(\sqrt{\log N \log \log N})}$ time and $\text{poly}(\log N)$ space. Further improvements and generalizations can be found in \cite{kuperberg-2}.

\begin{remark}
In \cite{kuperberg}, it is shown that HSP for $D_{2^t}$ reduces to determining the parity of $a$. 
\end{remark}

\subsection{Query complexity}

In \cite{ettinger-hoyer}, it is shown that a polynomial number of HSP samples is sufficient to recover $H_a$ using exponential time post-processing. A related result in \cite{ettinger-hoyer-2} using different methods shows the HSP problem in a general finite group has polynomial quantum query complexity.

Transposing $i \leftrightarrow j$, and applying a Hadamard gate to the state in \eqref{dcp-case}, gives the state
\begin{equation}
\frac{1}{\sqrt{2}}
\begin{pmatrix}
1 & 1 \\
1 & -1
\end{pmatrix}
\begin{pmatrix}
   \omega_N^{\alpha k} & \omega_N^{(a-\alpha)k}  \\
   \omega_N^{-(a-\alpha)k} & \omega_N^{-\alpha k}
  \end{pmatrix} = \frac{1}{\sqrt{2}} \begin{pmatrix}
  \omega_N^{\alpha k} (1 + \omega_N^{-ak}) & \omega_N^{-\alpha k} (1 + \omega_N^{ak} ) \\
  \omega_N^{\alpha k} (1 - \omega_N^{-ak}) & \omega_N^{-\alpha k} (1 - \omega_N^{ak} ) \\  
  \end{pmatrix}.
\end{equation}
The probability of observing the first row is 
\begin{equation}
  \frac{1}{2n} (1+\cos(2 \pi ak/N)) = \frac{1}{n} \cos^2(\pi ak/N).
\end{equation}
For the second row, it is
\begin{equation}
  \frac{1}{2n} (1-\cos(2 \pi ak/N)) = \frac{1}{n} \sin^2(\pi ak/N).
\end{equation}
We are now in the situation of \cite{ettinger-hoyer} and can apply the post-processing algorithm described (which is exponential in time) to determine $a$ with high probability, for large $N$.

\subsection{Relation to the subset sum problem} Given $x = (x_1, \ldots, x_m) \in (\Z/N\Z)^m$ and $r \in \Z/N\Z$, the problem of finding $b \in \left\{ 0, 1 \right\}^m$ such that $b \cdot x = r$ is called the subset sum problem over $\Z/N\Z$.

The vector $b$ corresponds to specifying a subset of the $x_1, \ldots, x_m$ that sum to $r$.  Denote by 
\begin{equation*}
  S_r^x = \left\{ b \in \left\{ 0, 1 \right\}^m : b \cdot x = r \right\}
\end{equation*}
the set of subset sums for $(x,r)$. 

If such a $b$ exists, then $(x,r)$ is called a legal instance. In the decision version of the subset sum problem, the problem is to determine whether a given $(x,r)$ is a legal instance. 

In \cite{regev}, it is shown that the ability to efficiently find an element $b \in S_r^x$ for a large fraction of legal instances gives an efficient algorithm to solve DHSP. Furthermore, \cite{childs} shows that the ability to quantum sample from $S_r^x$ allows one to efficiently implement an optimal measurement to determine $a$ from $m$ DCP samples.

The subset sum problem over $\Z$ is known to be an NP-complete problem. Since one can reduce the subset sum problem over $\Z$ to the subset sum problem over $\Z/N\Z$, by choosing a large enough modulus $N$, it follows that the subset sum problem over $\Z/N\Z$ is also NP-complete.

\subsection{Optimal measurements}

It is shown in \cite{ettinger-hoyer-3} that efficient elimination observables do not exist for the dihedral group. Further results can be found in \cite{childs}. In particular, let 
\begin{equation*}
    \nu = m/\log_2 N
\end{equation*}
be the density defined in \cite{childs}.

It is shown in \cite[Theorem 2]{childs} that if $\nu > 1 + 4/\log_2 N$, the probability of determining $a$ using the optimal measurement on $m$ DCP samples is $\ge 1/8$. Furthermore, for any $N$ and $m$, the probability of determining $a$ is 
\begin{equation}
\label{nu-bound}
  \le 2^m/N = 2^{(\nu-1) \log_2 N},
\end{equation}  
which is exponentially small in $\log_2 N$ for any fixed $\nu < 1$, and gives a trivial upper bound when $\nu \ge 1$.

More general results on optimal measurements to distinguish conjugate hidden subgroups in certain groups can be found in \cite{moore-russell}.

In \cite{childs}, the success probability of the optimal measurement is determined to
\begin{equation*}
    p_{m,N} = \frac{1}{2^m N^{m+1}} \sum_{x \in (\Z/N\Z)^m} \left( \sum_{r \in \Z/N\Z} \sqrt{\eta_r^x} \right)^2,
\end{equation*}
where $\eta_r^x := |S_r^x|$.

\begin{remark}
\label{dcp-actual-bound}
 For example, let $m = 2$, $N = 2^m$, and $\nu = 1$. Computer calculations show that $p_{m,N} \approx 0.6665$. On the other hand, we saw in Remark~\ref{cyclic-go} that we can achieve a success probability of $\approx 0.9892$ for $m = 2$ in the cyclic group case.
\end{remark}

In \cite{moore-russell}, it is shown that the optimal POVM measurement to determine $a$ from $m$ DCP samples exists and is given by the PGM. The theorem of Naimark states that a POVM measurement on a system can be realized by augmenting the system with ancilla registers, applying a unitary operator, and then a PVM measurement on the ancilla. Seen in this light, the result in \cite{ettinger-hoyer} implies that the success probability of the optimal measurement is $> 1 - \frac{1}{2N}$ if $\nu > 89$, though no efficient implementation is known.

\begin{remark}
  In the classical world, if we have a probabilistic algorithm that succeeds with probability $> \frac{1}{2}$, we can run the algorithm multiple times on the same input to make the success probability arbitrarily close to $1$. In the quantum world, we cannot in general reuse inputs which are quantum states, so running the quantum algorithm multiple times requires more quantum samples, unless one can clone the input samples. However, we will see in the last section that for some problems such as DCP, cloning the input samples is essentially equivalent to solving the original problem.
\end{remark}

\section{A probabilistic no-go result for DCP}

First, a unitary no-go result for DCP.
\begin{theorem}
\label{no-dcp}
  There is no unitary operation to compute the value of $a$ into a register from a list of DCP samples for $a$.
\end{theorem}
\begin{proof}
Suppose there is a unitary operator $U$ which has the effect
\begin{equation}
\label{a-compute}
    U | A \rangle |\psi^1_a \rangle \cdots | \psi^m_a \rangle |0 \rangle   = | \Sigma_a(\psi_a) \rangle | a \rangle
\end{equation}
for every $a$. That is, $U$ takes takes a list of DCP samples for fixed but unknown $a$, a blank initialization state $| 0 \rangle$, and an ancilla state $| A \rangle$, and then computes $a$ into the blank register.

For any other $b \not= a$, we must also have
\begin{equation}
\label{b-compute}
    U | A \rangle |\psi^1_b \rangle \cdots | \psi^m_b \rangle |0 \rangle   = | \Sigma_b(\psi_b) \rangle | b \rangle.
\end{equation}

There are choices of $\psi^i_c$ for $i = 1, \ldots, m$ such that
\begin{equation}
\label{dcp-frame}
   \langle \psi^i_a | \psi^i_b \rangle = \frac{1}{2}
\end{equation}
for all $a \not= b$ and $i = 1, \ldots m$. To see this, recall the states
\begin{align*}
    \psi_a = \frac{1}{\sqrt{2}} \left( | x^\alpha \rangle + | y x^{a-\alpha} \rangle \right), \\
    \psi_b = \frac{1}{\sqrt{2}} \left( | x^\beta \rangle + | y x^{b-\beta} \rangle \right), \\
\end{align*}
have possible inner product $\langle \psi_a | \psi_b \rangle \in \left\{ 0, \frac{1}{2}, 1 \right\}$, and there are choices of $\psi_a$ and $\psi_b$ such that
\begin{equation}
\label{nearly-orthogonal}
  \langle \psi_a | \psi_b \rangle \not= 0, 1,
\end{equation}  
for instance, if  $a \not= b$ and $a - \alpha = b - \beta$ or $\alpha = \beta$. In particular, taking
\begin{equation*}
  \psi^i_c = | x^c \rangle + | y x^0 \rangle
\end{equation*}
for $c \in \Z/N\Z$ satisfies \eqref{dcp-frame}.

Taking the inner product of \eqref{a-compute} and \eqref{b-compute}, we obtain
\begin{equation}
\label{normal-compute}
  \langle \psi^1_a | \psi^1_b \rangle \cdots \langle \psi^m_a | \psi^m_b \rangle  = \langle \Sigma_a(\psi_a) | \Sigma_b(\psi_b) \rangle \langle a | b \rangle = 0,
\end{equation}
a contradiction as we have shown there are choices of $\psi_a^i$ and $\psi_b^i$ making the left hand side of \eqref{normal-compute} non-zero.

\end{proof}
We will give yet another proof of Theorem~\ref{no-dcp} in Theorem~\ref{no-dcp-2}. The proof of Theorem \ref{no-dcp} mirrors the proof of the no cloning theorem \cite{wootters-zurek} and precludes unitary operations, but not more general quantum algorithms,  which may allow for approximate outputs, probabilistic processes, or post-processing. Indeed, computing the exact value of $a$ into a register is rather strong: even in the finite cyclic group case, the standard algorithm only determines a generator for the hidden subgroup using a process of the type given in Theorem~\ref{finite-cyclic-algorithm}.

The following is a probabilistic no-go result for DCP based on modifying the proof of the unitary no-go result for DCP. 
\begin{theorem}
\label{no-dcp-probabilistic-1}
  The problem of determining $a$, given a list of $m$ DCP samples for unknown $a$, is not unitarily solvable with a success parameter independent of $a$ that is $\ge 1 - \frac{1}{9} \cdot 2^{- 2m}$.
\end{theorem}
\begin{proof}
To ease notation, we let
\begin{align}
    \psi_a & = |\psi^1_a \rangle \cdots | \psi^m_a \rangle, \\
    \psi_b & = |\psi^1_b \rangle \cdots | \psi^m_b \rangle.
\end{align}

Suppose there is a unitary operator $U$ which has the effect
\begin{align}
\label{dcp-a}    U  | A \rangle |\psi^1_a \rangle \cdots | \psi^m_a \rangle |0 \rangle  & = | \Sigma_a(\psi_a) \rangle | a \rangle + \sum_{c \not= a} | \Sigma_c(\psi_a) \rangle | c \rangle, \\
\label{dcp-b}    U  | A \rangle |\psi^1_b \rangle \cdots | \psi^m_b \rangle |0 \rangle  & = | \Sigma_b(\psi_b) \rangle | b \rangle + \sum_{c \not= b} | \Sigma_c(\psi_b) \rangle | c \rangle, 
\end{align}
where 
\begin{align}
\label{dominant}
  |\Sigma_a(\psi_a)|^2 & \ge 1 - 2^{-\delta}, \\
  \notag |\Sigma_b(\psi_b)|^2 & \ge 1 - 2^{-\delta},
\end{align}
and $\delta$ is to be chosen.

Because of \eqref{dominant}, we have that
\begin{align}
    \sum_{c \not= a}|\Sigma_c(\psi_a)|^2 & < 2^{-\delta},  \\
    \notag \sum_{c \not= b} |\Sigma_c(\psi_b)|^2 & < 2^{-\delta}.
\end{align}
Taking the inner product of \eqref{dcp-a} and \eqref{dcp-b}, we obtain
\begin{align}
\label{main-ineq} \langle \psi^1_a | \psi^1_b \rangle \cdots \langle \psi^m_a | \psi^m_b \rangle 
  & \le | \langle \Sigma_a(\psi_a) | \Sigma_a(\psi_b) \rangle | + | \langle \Sigma_b(\psi_a) | \Sigma_b(\psi_b) \rangle | + \sum_{c \not= a,b } | \langle \Sigma_c(\psi_a)| \Sigma_c(\psi_b) \rangle| \\
\notag  & \le | \langle \Sigma_a(\psi_a) | \Sigma_a(\psi_b) \rangle | + | \langle \Sigma_b(\psi_a) | \Sigma_b(\psi_b) \rangle | + 2^{-\delta} \\
\notag  & \le 2^{-\delta} + 2\cdot 2^{-\delta/2} \\
\notag  & < 3 \cdot 2^{-\delta},
\end{align}
using Cauchy-Schwartz repeatedly. Arrange the left most side of \eqref{main-ineq} to be $2^{-m}$ as in \eqref{nearly-orthogonal} and we see that choosing $\delta \ge 2 \left( m + \log_2 3 \right)$ gives a contradiction to the above inequality.
\end{proof}

\begin{remark}
At fixed $\nu = m/\log_2 N$, Theorem~\ref{no-dcp-probabilistic-1} gives an upper bound on the success parameter of 
\begin{equation}
\label{nu-bound-2}
    1 - 2^{-2 (\nu \log_2 N + \log_2 3)} = 1 - \frac{1}{9} N^{-2\nu}.
\end{equation}
Although the bound in \eqref{nu-bound-2} seems far from optimal (see Remark~\ref{dcp-actual-bound}), it is still stronger than trivial bounds which result from \eqref{nu-bound} \cite[Theorem 2]{childs} or \cite{moore-russell} when $\nu \ge 1$.
\end{remark}

\section{Quantum cloning and DCP}

In this section, we explain a connection between DCP and quantum cloning. Although the topics in this section are not needed for the results of the previous section, the connection with quantum cloning helped motivate the proofs of the previous section, so we have included it for completeness.

By copying a state $| \psi \rangle$, we mean forming the composite state $| A \rangle | \psi \rangle | 0 \rangle  $ for a blank initialization state $| 0 \rangle$ and ancilla state $| A \rangle$, and applying a quantum algorithm to produce the state $| \Sigma(\psi) \rangle | \psi \rangle | \psi \rangle $. 

The no cloning theorem asserts that there is no unitary operation which can copy a general unknown quantum state. However, if the states are chosen from a known set of mutually orthogonal states, it is well known that cloning is possible, as shown for completeness in the following proposition.
\begin{proposition}
\label{copy-orthogonal}
Let $| \psi_{a;1} \rangle, \ldots, | \psi_{a;m} \rangle$ be a set of mutually orthogonal states which depend on a parameter $a$. Suppose $| \psi \rangle = | \psi_{a,i} \rangle$ for some index $i$ (which is unknown).  

If the value of $a$ is known, then there is a unitary operation which copies $| \psi \rangle$.
\end{proposition}
\begin{proof}
First note that we can copy any state $|i \rangle$ of the computational basis. Start with
\begin{equation*}
   | i \rangle | 0 \rangle =   | i_n \rangle \ldots | i_0 \rangle | 0 \rangle \ldots |0 \rangle,
\end{equation*}
where we have encoded the last two registers into $n$ qubits, for $n$ large enough.

Applying a CNOT gate to the $j$th and $(j+n+1)$th qubits $| i_j \rangle | 0 \rangle$  produces $| i_j \rangle | i_j \rangle$ for every $j$. Hence, we can produce the state
\begin{equation*}
  | i_n \rangle \ldots | i_0 \rangle | i_n \rangle \ldots | i_0 \rangle = |i \rangle | i \rangle.
\end{equation*}

Now, encode a unitary operator $U_a$ which has the effect
\begin{equation*}
   U_a | \psi_{a;i} \rangle = | i \rangle .
\end{equation*}
Starting with 
\begin{equation*}
     |\psi_{a;i} \rangle | 0 \rangle ,
\end{equation*}
apply $U_a$ to the first register to obtain
\begin{equation*}
    | i \rangle | 0 \rangle.
\end{equation*}
Copy the state $| i \rangle$ to obtain
\begin{equation*}
  | i \rangle | i \rangle.
\end{equation*}
Applying $U_a^{-1}$ to both registers gives
\begin{equation*}
   | \psi_{a;i} \rangle | \psi_{a;i} \rangle.
\end{equation*}
\end{proof}

Later we will need a slightly stronger version of Proposition~\ref{copy-orthogonal}.
\begin{proposition}
\label{copy-orthogonal2}
Let $| \psi_{a;1} \rangle, \ldots, | \psi_{a;m} \rangle$ be a set of mutually orthogonal states which depend on a parameter $a$ and assume we can encode a unitary operator $T$ such that $T | a \rangle |\psi_{a;i} \rangle = |a \rangle | i \rangle$. 

Suppose $| \psi \rangle = | \psi_{a; i} \rangle$ for some index $i$ (which is unknown). If we have the value of $a$ in a register, then there is a unitary operation which copies $| \psi \rangle$.
\end{proposition}
\begin{proof}
Starting with 
\begin{equation*}
     | a \rangle  |\psi_{a,i} \rangle | 0 \rangle | 0 \rangle ,
\end{equation*}
apply $T$ to obtain
\begin{equation*}
    |a \rangle | i \rangle | 0 \rangle | 0 \rangle.
\end{equation*}
Copy the states  $|a \rangle$ and $| i \rangle$ to obtain
\begin{equation*}
        |a \rangle | i \rangle | a \rangle | i \rangle.
\end{equation*}
Applying $T^{-1}$ to both pairs of registers gives
\begin{equation*}
   |a \rangle | \psi_{a,i} \rangle |a \rangle | \psi_{a,i} \rangle.
\end{equation*}
which we can permute to obtain
\begin{equation*}
   |a \rangle | \psi_{a,i} \rangle | \psi_{a,i} \rangle | a \rangle.
\end{equation*}
\end{proof}

\begin{proposition}
\label{dcp-cloning}
If we can copy any given DCP sample
\begin{equation}
\label{dcp-sample}
  \psi_{a; \alpha} = \frac{1}{\sqrt{2}} \left( | x^\alpha \rangle + | y x^{a-\alpha} \rangle \right),
\end{equation}
to produce a state of the form
\begin{equation}
\label{dcp-equal}
  \psi_{a; \alpha} \otimes \psi_{a; \alpha} = \frac{1}{\sqrt{2}} \left( | x^\alpha \rangle  + | y x^{a-\alpha} \rangle  \right) \otimes \frac{1}{\sqrt{2}} \left( | x^\alpha \rangle  + | y x^{a-\alpha} \rangle  \right),
\end{equation}
then we can determine the value of $a$ from DCP samples for $a$.

If $a$ is known, then we can copy any given DCP sample for $a$ using a unitary operation.
\end{proposition}
\begin{proof}
Given samples of the form \eqref{dcp-equal}, we measure both registers, and with probability $1/2$ we obtain 
\begin{equation}
  | x^\alpha \rangle | y x^{a-\alpha} \rangle \text{ or } |y x^{a-\alpha} \rangle | x^\alpha \rangle.
\end{equation}
The sum of the observed exponents of the two registers gives $a$.

If $a$ is known, then DCP samples for $a$,
\begin{equation*}
  \psi_{a; \alpha} = \frac{1}{\sqrt{2}} \left( |x^\alpha \rangle  + | y x^{a-\alpha} \rangle \right),
\end{equation*}
are chosen from a set of mutually orthogonal states depending on the parameter $a$. By Proposition~\ref{copy-orthogonal}, for each sample of the form \eqref{dcp-sample} we can copy it to produce a sample of the form \eqref{dcp-equal}.
\end{proof}

\begin{remark}
Copying a DCP sample up to parity would allow one to determine the parity of $a$ and vice versa.
\end{remark}

\begin{theorem}
\label{no-dcp-clone}
  If $a$ is unknown, there is no unitary operation, which from a list of DCP samples for $a$, copies an additional DCP sample for the same $a$, while leaving the list of DCP samples alone.
\end{theorem}
\begin{proof} 
Suppose there is a unitary operator $U$ which transforms
\begin{equation}
\label{clone-a}
  U | A \rangle | \psi^1_a \rangle  \cdots | \psi^m_a \rangle | \psi_a \rangle | 0 \rangle  = | \Sigma_a(\psi_a) \rangle |\psi^1_a \rangle \cdots | \psi^m_a \rangle | \psi_a \rangle | \psi_a \rangle ,
\end{equation}
where $\psi_a = \psi_{a;\alpha} = \frac{1}{\sqrt{2}} (|x^\alpha \rangle + | y x^{a-\alpha} \rangle)$ is a DCP sample for $a$ fixed, and $\alpha$ randomly chosen for each such state.
We are supposing $U$ performs the above operation for {\it any} (unknown) $a$. Thus, we also have that
\begin{align}
\label{clone-b}
  U |A \rangle | \psi^1_b \rangle  \cdots | \psi^m_b\rangle | \psi_b \rangle | 0 \rangle  = | \Sigma_b(\psi_b) \rangle |\psi^1_b \rangle \cdots | \psi^m_b \rangle | \psi_b \rangle | \psi_b \rangle .
\end{align}
for any other $b$.

Taking the inner product of both sides of \eqref{clone-a} and \eqref{clone-b} we deduce 
\begin{equation}
\label{normal}
  \langle \psi^1_a | \psi^1_b \rangle \cdots \langle \psi^m_a | \psi^m_b \rangle \langle \psi_a | \psi_b \rangle = \langle \psi^1_a | \psi^1_b \rangle \cdots \langle \psi^m_a | \psi^m_b \rangle \langle \psi_a | \psi_b \rangle^2  \langle \Sigma_a(\psi_a) | \Sigma_b(\psi_b) \rangle.
\end{equation}
However, there are choices of $\psi^i_a, \psi^i_b$ for $i = 1, \ldots, m$, and $\psi_a, \psi_b$ which do not satisfy \eqref{normal} from \eqref{nearly-orthogonal}.

We may thus suppose without loss of generality that $\langle \psi^i_a | \psi^i_b \rangle \not= 0,1$ for all $i = 1, \ldots, N$, and hence \eqref{normal} becomes
\begin{equation*}
  \langle \psi_a | \psi_b \rangle = \langle \psi_a | \psi_b \rangle^2 \langle \Sigma_a(\psi_a) | \Sigma_b(\psi_b) \rangle.
\end{equation*}
We obtain a contradiction again by choosing $\psi_a$ and $\psi_b$ so that $\langle \psi_a | \psi_b \rangle \not= 0, 1$ as then
\begin{align}
  | \langle \psi_a | \psi_b \rangle | = \frac{1}{2}, \\
  | \langle \psi_a | \psi_b \rangle^2 \langle \Sigma_a(\psi_a) | \Sigma_b(\psi_b) \rangle  | \le \frac{1}{4}.
\end{align}
\end{proof}

The following is another proof of Theorem~\ref{no-dcp} using the connection with quantum cloning.
\begin{theorem}
\label{no-dcp-2}
  There is no unitary operation to compute the value of $a$ into a register from a list of DCP samples for $a$.
\end{theorem}
\begin{proof}
Suppose there is a unitary operator $U$ which has the effect
\begin{equation}
    U | A \rangle |\psi^1_a \rangle \cdots | \psi^m_a \rangle |0 \rangle   = | \Sigma_a(\psi_a) \rangle | a \rangle .
\end{equation}
That is, $U$ takes takes a list of DCP samples for fixed but unknown $a$, a blank initialization state $| 0 \rangle$, and an ancilla state $| A \rangle$, and then computes $a$ into the blank register.

Using an additional blank register and copying $| a \rangle$, there is a unitary operator $U'$ with the effect
\begin{equation}
    U' | A \rangle |\psi^1_a \rangle \cdots | \psi^m_a \rangle |0 \rangle |0 \rangle   = | \Sigma_a(\psi_a) | a \rangle |a \rangle .
\end{equation}
Use $U^{-1}$ and permute $| a \rangle$ and $| 0 \rangle$  to obtain
\begin{equation}
   | A \rangle |\psi^1_a \rangle \cdots | \psi^m_a \rangle  |a \rangle |0 \rangle .
\end{equation}

Thus, without loss of generality, we may assume the unitary operator $U$ has the effect
\begin{equation*}
    U | A \rangle |\psi^1_a \rangle \cdots | \psi^m_a \rangle |0 \rangle   = | A \rangle | \psi^1_a \rangle \cdots | \psi^m_a \rangle | a \rangle .
\end{equation*}
That is, $U$ takes takes a list of DCP samples for fixed but unknown $a$, a blank initialization state $| 0 \rangle$, and an ancilla state $| A \rangle$, and then computes $a$ into the blank register, while leaving the list of DCP samples alone.

Now, note that DCP samples $\psi_{a;\alpha}$ can be encoded using two registers as 
\begin{equation*}
  \frac{1}{\sqrt{2}} ( | 0 \rangle | \alpha  \rangle + |1 \rangle | a - \alpha \rangle).
\end{equation*}
The unitary operator $V$ which sends 
\begin{align*}
  V | a \rangle | 0 \rangle | \alpha \rangle & = | a \rangle | 0 \rangle | \alpha \rangle, \\
  V | a \rangle | 1 \rangle | \alpha \rangle & = | a \rangle | 1 \rangle | a - \alpha \rangle,
\end{align*}
will have the effect
\begin{equation*}
  V | a \rangle | \psi_{a;\alpha} \rangle = | a \rangle  \frac{1}{\sqrt{2}} \left( |0 \rangle  + | 1 \rangle \right)  | \alpha \rangle.
\end{equation*}
Using a Hadamard gate, we can encode a unitary operator $U_0$ such that
\begin{align*}
  U_0  \frac{1}{\sqrt{2}} \left( |0 \rangle  + | 1 \rangle \right) | \alpha \rangle = | 0 \rangle | \alpha \rangle \\
  U_0  \frac{1}{\sqrt{2}} \left( |0 \rangle  - | 1 \rangle \right) | \alpha \rangle = | 1 \rangle | \alpha \rangle \\
\end{align*}
Then the unitary operator $(I \otimes U_0) V$ has the effect
\begin{equation*}
  (I \otimes U_0) V | a \rangle | \psi_{a;\alpha} \rangle = | a \rangle |0 \rangle | \alpha \rangle. 
\end{equation*}

We can thus apply Proposition~\ref{copy-orthogonal2} to copy an additional DCP sample for the same $a$ using a unitary operation, while leaving the list of DCP samples alone. This contradicts Theorem~\ref{no-dcp-clone}.
\end{proof}



\end{document}